\documentclass[12pt]{amsart}
\usepackage{amssymb,latexsym}
\usepackage{epsfig}

\newtheorem{theorem}{Theorem}
\newtheorem{definition}{Definition}

\newtheorem{proposition}{Proposition}
\newtheorem{lemma}{Lemma}

\newcommand{\vp}{\varphi}
\newcommand{\imag}{\text{Im\,}}
\newcommand{\real}{\text{Re\,}}
\newcommand{\Cbb}{\mathbb{C}}
\newcommand{\Rbb}{\mathbb{R}}
\newcommand{\Nbb}{\mathbb{N}}
\newcommand{\sech}{\text{sech}}

\newcommand{\mgap}{\vspace{24pt}}
\newcommand{\sgap}{\vspace{12pt}}

\begin{document}

\title[The Pseudospectrum of the  Zakharov-Shabat System]{The Pseudospectrum\\
    of the Zakharov-Shabat System}
\author{Michael VanValkenburgh}
\maketitle

\begin{abstract}
	We study the pseudospectrum of the non-selfadjoint Zakharov-Shabat system in the semiclassical regime. The pseudospectrum may be defined as the union of the spectra of perturbations of the Zakharov-Shabat system, thus it is relevant to the numerical computation of true eigenvalues.
\end{abstract}

\mgap

\section{Introduction}

The Zakharov-Shabat system is the non-selfadjoint system of
first-order differential equations given by
\begin{equation}\label{E:tradZS}
    h\partial_{x}
    \left(
    \begin{matrix}
        u_{1}\\
        u_{2}
    \end{matrix}
    \right)
    =
    \left(
    \begin{matrix}
        -i\lambda &A(x)e^{iS(x)/h}\\
        -A(x)e^{-iS(x)/h}&i\lambda
    \end{matrix}
    \right)
    \left(
    \begin{matrix}
        u_{1}\\
        u_{2}
    \end{matrix}
    \right)
\end{equation}
where $A>0$ and $S$ are real-valued functions, $h>0$ is the
semiclassical parameter, and $\lambda$ is the (complex) spectral
parameter. As is well-known, Zakharov and Shabat \cite{R:ZS72} found this system to be one half of the Lax pair for the focusing nonlinear Schr\"{o}dinger equation
$$ih\partial_{t}\psi+\tfrac{1}{2}h^{2}\partial_{x}^{2}\psi+|\psi|^{2}\psi=0,\qquad \psi(x,0)=A(x)e^{iS(x)/h}.$$
Writing $v_{1}=e^{-iS/2h}u_{1}$, $v_{2}=e^{iS/2h}u_{2}$, and writing $D_{x}=\frac{1}{i}\partial_{x}$, we put the system (\ref{E:tradZS}) into the form:
\begin{equation}\label{E:ZS}
    \left(
    \begin{matrix}
        -hD_{x}-\frac{1}{2}S'(x) &-iA(x)\\
        -iA(x)&hD_{x}-\frac{1}{2}S'(x)
    \end{matrix}
    \right)
    \left(
    \begin{matrix}
        v_{1}\\
        v_{2}
    \end{matrix}
    \right)
    =
    \lambda
    \left(
    \begin{matrix}
        v_{1}\\
        v_{2}
    \end{matrix}
    \right).
\end{equation}
We denote the operator on the left-hand side by $P$, having
principal symbol
\begin{equation*}
    p(x,\xi)=
    \left(
    \begin{matrix}
        -\xi-\frac{1}{2}S'(x) &-iA(x)\\
        -iA(x) &\xi-\frac{1}{2}S'(x)
    \end{matrix}
    \right).
\end{equation*}

\sgap

In this paper we study the pseudospectrum of $P$, the
set where the resolvent of $P$ is large. Equivalently, the pseudospectrum of $P$ may be defined as the union of the spectra of \emph{perturbations} of $P$ \cite{R:TrefEmbree}. Thus the pseudospectrum is relevant, for example, to the recent numerical experiments of Kim, Lee, and Lyng which suggest $\mathcal{O}(h^{2})$ convergence of the WKB eigenvalues to the true eigenvalues in the semiclassical limit $h\to 0$ \cite{R:KimLeeLyng}. (They restrict to the case when $S\equiv 0$ and $A$ is even, bell-shaped, and real analytic.)

\sgap

We use a standard method of microlocal analysis: we show that if a certain Poisson bracket condition is satisfied, appearing as a condition on $P$ and on the spectral parameter $\lambda$, then we can explicitly construct quasimodes, starting from a complex geometrical optics ansatz. This method was used by H\"{o}rmander \cite{R:HoWithoutSolns} and was rediscovered by Davies \cite{R:Davies}, as observed by Zworski \cite{R:ZwRemark}. Extensions of this method may be found in the papers of Dencker, Sj\"{o}strand, and Zworski \cite{R:DenckerPSsys}, \cite{R:DSZ}, which strongly influenced the work presented here.

\mgap

The main result of this paper is the following:
\begin{theorem}\label{T:mainQM}
    Let $S\in C^{\infty}(\Rbb;\Rbb)$, let
    $A\in\mathcal{S}(\Rbb;\Rbb)$ (that is, a Schwartz function), $A>0$, and let $\lambda\in\Cbb$
    be such that for some $x_{0}\in\Rbb$ we have
    \begin{equation*}
        \real\lambda=-\tfrac{1}{2}S'(x_{0}) \quad\text{and}\quad
        0<|\imag\lambda|<A(x_{0}).
    \end{equation*}
    Moreover, assume that $S^{(2k)}(x_{0})\neq 0$ is the first
    nonvanishing derivative of $S$, at $x_{0}$, of order $\geq 2$ (so $k\geq
    1$). Then there exists $h_{0}>0$, and for any $N\in\Nbb$ there
    exists $u_{N}=u_{N}(\cdot\,;h)\in C^{\infty}_{0}(\Rbb)$ with $||u_{N}||_{L^{2}}=1$ and some constant
    $C_{N}>0$ such that
    \begin{equation*}
        ||(P(x,hD_{x})-\lambda I)u_{N}||_{L^{2}}\leq
        C_{N}h^{N}\qquad\forall h\in(0,h_{0}).
    \end{equation*}
\end{theorem}

\sgap

Here it is most practical to state the result in terms of
derivatives of $S$. However, as emphasized by Dencker,
Sj\"{o}strand, and Zworski \cite{R:DenckerPSsys}, \cite{R:DSZ},
the underlying general mechanisms are the repeated Poisson
brackets of the real and imaginary parts of $d(x,\xi)$, defined as
\begin{equation*}
    d(x,\xi):=\det (p(x,\xi)-\lambda I).
\end{equation*}
Indeed,
\begin{equation*}
    \begin{aligned}
    d(x,\xi)=-\xi^{2}+(\tfrac{1}{2}S'(x)+\real\lambda)^{2}-(\imag\lambda)^{2}
    +A(x)^{2}+2i(\tfrac{1}{2}S'(x)+\real\lambda)\imag\lambda
    \end{aligned}
\end{equation*}
and
\begin{equation*}
    \{\real d,\imag d\}=-2\xi S''(x)\imag\lambda.
\end{equation*}
The general formulas for higher Poisson brackets are rather messy, but
the first nonvanishing Poisson bracket takes a simple form. Let
$S^{(k+1)}$ be the first nonvanishing derivative of $S$ of order
greater than or equal to two ($k\geq 1$). Then
\begin{equation*}
    H^{k}_{\real d}\imag d:=\{\real d,\{\real d,\{\ldots,\imag d\}\}\ldots\}=(-2\xi)^{k}S^{(k+1)}(x)\imag\lambda
\end{equation*}
and all other Poisson brackets of order $\leq k$ are equal to
zero.

\sgap

In Sections \ref{S:GOansatz} through \ref{S:FinalEsts} we prove
Theorem \ref{T:mainQM}, constructing quasimodes and thus proving
blow-up of the resolvent as $h\to 0$. On the other hand, in Sections
\ref{S:UpperResolv} and \ref{S:geommeaning} we consider upper
bounds for the resolvent. We prove that the genuine spectrum is
discrete off the real line in Sections \ref{S:Sigmas} and
\ref{S:Discreteness}.

\sgap

It remains to be seen what happens when the first nonvanishing
derivative of $S$ is an odd derivative. We would expect to have a
subelliptic estimate; hopefully in the future we
can do something concrete and fairly simple for the case of the
Zakharov-Shabat operator. It may also be interesting to more
carefully study the boundary of the pseudospectrum. For both of
these issues, we would welcome further physically significant
examples from the physics community.

\sgap

\noindent\textbf{Acknowledgement:} The author thanks M. Hitrik for
suggesting the problem and for helpful conversations.

\mgap

\section{The Geometrical Optics Ansatz}\label{S:GOansatz}

We take the geometrical optics ansatz,
\begin{equation}\label{E:GOansatz}
    e^{i\vp(x)/h}
    \left(
    \begin{matrix}
        a(x;h)\\
        b(x;h)
    \end{matrix}
    \right),
\end{equation}
with
\begin{equation*}
    \begin{aligned}
    a(x;h)&=a_{0}(x)+ha_{1}(x)+h^{2}a_{2}(x)+\cdots\qquad\text{and}\\
    b(x;h)&=b_{0}(x)+hb_{1}(x)+h^{2}b_{2}(x)+\cdots,
    \end{aligned}
\end{equation*}
and we let
\begin{equation*}
    M(x):=
    \left(
    \begin{matrix}
        \vp'+\frac{1}{2}S'+\lambda &iA\\
        iA  &-\vp'+\frac{1}{2}S'+\lambda
    \end{matrix}
    \right).
\end{equation*}
Then for the ansatz (\ref{E:GOansatz}) to formally solve
(\ref{E:ZS}), we group terms in the same order of $h$ and thus want
\begin{equation}\label{E:Mzero}
    M
    \left(
    \begin{matrix}
        a_{0}\\
        b_{0}
    \end{matrix}
    \right)=
    \left(
    \begin{matrix}
        0\\
        0
    \end{matrix}
    \right),\qquad \text{and}
\end{equation}
\begin{equation}\label{E:Mrest}
    M
    \left(
    \begin{matrix}
        a_{j+1}\\
        b_{j+1}
    \end{matrix}
    \right)=
    \left(
    \begin{matrix}
        ia_{j}'\\
        -ib_{j}'
    \end{matrix}
    \right)
    \qquad \forall j\in\{0,1,2,3,\cdots\}.
\end{equation}

\mgap

\section{The Eikonal Equation}\label{S:Eikonal}

In order to have non-zero solutions to (\ref{E:Mzero}), of course
we need to have $\det M=0$; that is, we need $\vp$ to solve the
(complex) eikonal equation:
\begin{equation*}
    (\vp')^{2}=(\tfrac{1}{2}S' +\lambda)^{2}+A^{2}.
\end{equation*}
We note that the turning-point curve, defined to be the set
where $\vp'=0$, is given parametrically by
\begin{equation*}
    \lambda(x)=-\tfrac{1}{2}S'(x)\pm iA(x).
\end{equation*}
While it is possible for our $\lambda$ to lie on the turning point
curve (for a point other than $x_{0}$), we still have that $\vp'(x_{0})\neq 0$ by the hypothesis on $\lambda$.

Near the point $x_{0}$, where we choose to take $\vp(x_{0})=0$,
this has the solution
\begin{equation*}
    \vp(x)=\pm\int^{x}_{x_{0}}\sqrt{(\tfrac{1}{2}S'(t)+\lambda)^{2}+A(t)^{2}}\,
    dt.
\end{equation*}
Taking Taylor expansions and integrating, we get
\begin{equation*}
    \begin{aligned}
    \vp(x)&=\pm\left[\sqrt{(\tfrac{1}{2}S'(x_{0})+\lambda)^{2}+A(x_{0})^{2}}\right](x-x_{0})\\
    &\qquad
    \pm\frac{1}{4}\left[\frac{
    S''(x_{0})(\frac{1}{2}S'(x_{0})
    +\lambda)+2A(x_{0})A'(x_{0})}{\sqrt{(\frac{1}{2}S'(x_{0})+\lambda)^{2}+A(x_{0})^{2}}}\right]
    (x-x_{0})^{2} +\mathcal{O}((x-x_{0})^{3}).
    \end{aligned}
\end{equation*}

To prove the theorem in its full generality, we will need to
expand $\vp$ to higher orders. But this is simplified by the fact
that, in the final estimates, the important object is the
\emph{imaginary} part of the phase $\vp$. For this we have the
following lemma, where for convenience we let
$$\alpha:=(\tfrac{1}{2}S'(x_{0})+\lambda)^{2}+A(x_{0})^{2}=A(x_{0})^{2}-(\imag\lambda)^{2}\quad (>0).$$

\sgap

\begin{lemma}\label{L:imagphi}
Let $m$ be the order of the first nonvanishing derivative of $S$
at the point $x_{0}$. Then
$\imag\vp(x)=\frac{\pm\imag\lambda}{2(m!)\sqrt{\alpha}}S^{(m)}(x_{0})(x-x_{0})^{m}+\mathcal{O}((x-x_{0})^{m+1})$.
\end{lemma}
\begin{proof}
    Let $T^{S}_{k}$ denote the $k$th Taylor coefficient, centered
    at $x_{0}$, of $(\frac{1}{2}S'(x)+\lambda)^{2}$, and let
    $T^{A}_{k}$ denote that of $A(x)^{2}$. We then have, for $x$ sufficiently near $x_{0}$,
    \begin{equation*}
        \begin{aligned}
        &\pm\sqrt{(\tfrac{1}{2}S'(x)+\lambda)^{2}+A(x)^{2}}\\
        &\qquad=\sum_{n=0}^{\infty}
        {1/2\choose
        n}\alpha^{\frac{1}{2}-n}
        \left(\sum_{k=1}^{\infty}T^{S}_{k}(x-x_{0})^{k}+\sum_{k=1}^{\infty}T_{k}^{A}(x-x_{0})^{k}\right)^{n}.
        \end{aligned}
    \end{equation*}
    As in the statement of the theorem, let $S^{(j+1)}$ be the smallest nonvanishing derivative of $S$
    of order greater than or equal to two ($j\geq 1$). Then
    \begin{equation*}
        \begin{aligned}
        (j!)T^{S}_{j}
        &=(\tfrac{1}{2}S'(x_{0})+\lambda)S^{(j+1)}(x_{0})\\
        &=i(\imag\lambda)S^{(j+1)}(x_{0})\\
        &\neq 0\qquad \text{since }\imag\lambda\neq 0.
        \end{aligned}
    \end{equation*}
    Hence
    \begin{equation*}
        \begin{aligned}
        &\pm\imag\sqrt{(\tfrac{1}{2}S'(x)+\lambda)^{2}+A(x)^{2}}\\
        &\qquad=\imag\sum_{n=0}^{\infty}{1/2\choose n}\alpha^{\frac{1}{2}-n}
        \left(\sum_{k=j}^{\infty}T^{S}_{k}(x-x_{0})^{k}+\sum_{k=1}^{\infty}T_{k}^{A}(x-x_{0})^{k}\right)^{n}\\
        &\qquad=\imag\left[\tfrac{1}{2}\alpha^{-1/2}\,T^{S}_{j}(x-x_{0})^{j}+\mathcal{O}((x-x_{0})^{j+1})\right]\\
        &\qquad=\frac{\imag\lambda}{2\sqrt{\alpha}(j!)}S^{(j+1)}(x_{0})(x-x_{0})^{j}+\mathcal{O}((x-x_{0})^{j+1}).
        \end{aligned}
    \end{equation*}
    And so by integrating we finally get
    \begin{equation*}
        \imag\vp(x)
        =\frac{\pm\imag\lambda}{2\sqrt{\alpha}(j+1)!}S^{(j+1)}(x_{0})(x-x_{0})^{j+1}+\mathcal{O}((x-x_{0})^{j+2}),
    \end{equation*}
    which proves the lemma.
\end{proof}

\mgap

\section{The Transport Equations}\label{S:Transport}
Since we are taking $\vp$ to solve the eikonal equation, for general $x$ the image
of $M(x)$ is spanned by the eigenvector
\begin{equation*}
    \left(
    \begin{matrix}
    \vp'+\frac{1}{2}S'+\lambda\\
    iA
    \end{matrix}
    \right)
\end{equation*}
having eigenvalue $S'+2\lambda$. Since this eigenvalue is not
zero and since $\vp'+\frac{1}{2}S'+\lambda\neq 0$, we can diagonalize $M$ as follows:
\begin{equation*}
    M=R
    \left(
    \begin{matrix}
    S'+2\lambda &0\\
    0 &0
    \end{matrix}
    \right)
    R^{-1}
\end{equation*}
where
\begin{equation*}
    R=
    \left(
    \begin{matrix}
    \vp'+\frac{1}{2}S'+\lambda &-iA\\
    iA &\vp'+\frac{1}{2}S'+\lambda
    \end{matrix}
    \right)
\end{equation*}
and hence
\begin{equation*}
    R^{-1}=(S' +2\lambda)^{-1}
    \left(
    \begin{matrix}
    1 &\frac{iA}{\vp'+\frac{1}{2}S'+\lambda}\\
    \frac{-iA}{\vp'+\frac{1}{2}S'+\lambda} &1
    \end{matrix}
    \right).
\end{equation*}
Then, writing (\ref{E:Mzero}) and (\ref{E:Mrest}) in terms of this
diagonalization, we want $a$ and $b$ to satisfy
\begin{equation}\label{E:ablink}
    a_{j}+\left(\frac{iA}{\vp'+\frac{1}{2}S'+\lambda}\right)b_{j}=
    \begin{cases}
    0 \qquad &\text{if }j=0\\
    \frac{ia_{j-1}'}{S'+2\lambda}+\frac{Ab_{j-1}'}{(S'+2\lambda)(\vp'+\frac{1}{2}S'+\lambda)}\qquad
    &\text{if }j\geq 1.
    \end{cases}
\end{equation}
and
\begin{equation}\label{E:abderiv}
    a_{j}'=\frac{i(\vp'+\frac{1}{2}S'+\lambda)b_{j}'}{A}
    \qquad\qquad \forall j\in\{0,1,2,\ldots\}.
\end{equation}

\sgap

We will now construct $a_{0}$ and $b_{0}$ in detail. First of all,
we want
\begin{equation*}
    \left(
    \begin{matrix}
        a_{0}\\
        b_{0}
    \end{matrix}
    \right)\in
    \text{Ker(M)}=\text{Span}\left\{
    \left(
    \begin{matrix}
        -iA\\
        \vp'+\frac{1}{2}S'+\lambda
    \end{matrix}
    \right)
    \right\}.
\end{equation*}
And secondly, we want $\left(
\begin{smallmatrix}
ia_{0}'\\
-ib_{0}'
\end{smallmatrix}
\right)$ to be in the image of $M$. Therefore, we want both
\begin{equation*}
    \left(
    \begin{matrix}
        a_{0}\\
        b_{0}
    \end{matrix}
    \right)
    =
    \alpha(x)
    \left(
    \begin{matrix}
        -iA\\
        \vp'+\frac{1}{2}S'+\lambda
    \end{matrix}
    \right)
\end{equation*}
and
\begin{equation*}
    \left(
    \begin{matrix}
        ia_{0}'\\
        -ib_{0}'
    \end{matrix}
    \right)
    =
    \beta(x)
    \left(
    \begin{matrix}
        \vp'+\frac{1}{2}S'+\lambda\\
        iA
    \end{matrix}
    \right),
\end{equation*}
where the coefficients $\alpha$ and $\beta$ are to be determined.
But by a direct calculation, this is possible when
\begin{equation*}
    \alpha(x)=\left[(\vp'(x)+\tfrac{1}{2}S'(x)+\lambda)^{2}+A(x)^{2}\right]^{-1/2},
\end{equation*}
which gives us $a_{0}$ and $b_{0}$.

\sgap

To solve for the remaining amplitudes, for $j\geq 1$ in
(\ref{E:ablink}) and (\ref{E:abderiv}), we let
\begin{equation*}
    \gamma:=\frac{iA}{\vp'+\frac{1}{2}S'+\lambda}
\end{equation*}
and
\begin{equation*}
    c_{j-1}:=\frac{ia_{j-1}'}{S'+2\lambda}+\frac{Ab_{j-1}'}{(S'+2\lambda)(\vp'+\frac{1}{2}S'+\lambda)}.
\end{equation*}
(Note that $\gamma-\frac{1}{\gamma}=\frac{2i\vp'}{A}\neq 0$ for $x$ near $x_{0}$.) Then we
are to solve the system
\begin{equation*}
    \begin{cases}
    a_{j}+\gamma b_{j}=c_{j-1}\\
    a_{j}'+\frac{1}{\gamma}b_{j}'=0.
    \end{cases}
\end{equation*}
But this is easily accomplished.

\mgap

\section{The Final Estimates}\label{S:FinalEsts}

It is now time to complete the quasimode construction by
estimating the error generated by taking only finitely many terms
in (\ref{E:GOansatz}), hence making rigorous the asymptotic
series.

We take only finitely many amplitude terms:
\begin{equation}\label{E:finiteamp}
    \begin{aligned}
    a(x;h)&=a_{0}(x)+ha_{1}(x)+h^{2}a_{2}(x)+\cdots+h^{N}a_{N}(x)\qquad\text{and}\\
    b(x;h)&=b_{0}(x)+hb_{1}(x)+h^{2}b_{2}(x)+\cdots+h^{N}b_{N}(x).
    \end{aligned}
\end{equation}
Then
\begin{equation*}
    \begin{aligned}
    (P-\lambda I)
    \left(
    \begin{matrix}
    e^{i\vp/h}a\\
    e^{i\vp/h}b
    \end{matrix}
    \right)
    &=e^{i\vp/h}\left[ih^{N+1}
    \left(
    \begin{matrix}
    a_{N}'\\
    -b_{N}'
    \end{matrix}
    \right)
    +\sum_{k=0}^{N-1}h^{k+1}\left(i
    \left(
    \begin{matrix}
    a_{k}'\\
    -b_{k}'
    \end{matrix}
    \right)
    -M
    \left(
    \begin{matrix}
    a_{k+1}\\
    b_{k+1}
    \end{matrix}
    \right)
    \right)
    -M
    \left(
    \begin{matrix}
    a_{0}\\
    b_{0}
    \end{matrix}
    \right)
    \right]\\
    &=ih^{N+1}e^{i\vp/h}
    \left(
    \begin{matrix}
    a_{N}'\\
    -b_{N}'
    \end{matrix}
    \right)
    \end{aligned}
\end{equation*}
where we have solved the eikonal and transport equations as above.

\sgap

We now assume that $S^{(2k)}(x_{0})\neq 0$ is the first
nonvanishing derivative of $S$ of order greater than or equal to
two ($k\geq 1$). Then, using Lemma \ref{L:imagphi} with $m=2k$, we
choose the sign of $\vp$ such that the leading term
\begin{equation*}
    \frac{\pm\imag\lambda}{2(2k)!\sqrt{\alpha}}S^{(2k)}(x_{0})(x-x_{0})^{2k}
\end{equation*}
is a nonnegative quantity. Then there exists some $\gamma>0$ such
that, for $x$ sufficiently close to $x_{0}$,
\begin{equation*}
    \gamma(x-x_{0})^{2k}\leq \imag\vp(x)\leq 3\gamma(x-x_{0})^{2k}.
\end{equation*}

\sgap

To conclude the quasimode construction, we let $\chi\in
C^{\infty}_{0}(\Rbb)$ be $=1$ for $|x-x_{0}|<\frac{1}{2}\delta$
and $=0$ for $|x-x_{0}|>\delta$, where $\delta>0$ is to be
determined. Then we set
\begin{equation*}
    f(x)=e^{i\vp/h}
    \left(
    \begin{matrix}
    a\\
    b
    \end{matrix}
    \right)
\end{equation*}
with $a$ and $b$ as in (\ref{E:finiteamp}). And we let
\begin{equation*}
    \tilde{f}(x)=\chi(x)f(x).
\end{equation*}
Then
\begin{equation*}
    ||(P-\lambda I)\tilde{f}||_{2}\leq \left|\left|
    \left(
    \begin{matrix}
    ih\chi' & 0\\
    0& -ih\chi'
    \end{matrix}
    \right)
    f
    \right|\right|_{2}
    +||\chi (P-\lambda I)f||_{2}.
\end{equation*}
As already noted,
\begin{equation*}
    ||\chi (P-\lambda I)f||_{2}=h^{N+1}
    \left|\left|\chi e^{i\vp/h}
    \left(
    \begin{matrix}
    a_{N}'\\
    -b_{N}'
    \end{matrix}
    \right)\right|\right|_{2},
\end{equation*}
and then we compute
\begin{equation*}
    \begin{aligned}
    \left|\left|\chi e^{i\vp/h}
    \left(
    \begin{matrix}
    a_{N}'\\
    -b_{N}'
    \end{matrix}
    \right)\right|\right|^{2}_{2}
    &=\int|\chi|^{2}e^{-2\imag\vp/h}[|a_{N}'|^{2}+|b_{N}'|^{2}]\,
    dx\\
    &\leq C\int_{|x-x_{0}|\leq
    \delta}e^{-2\imag\vp/h}\,dx\\
    &\leq C\int_{|x-x_{0}|\leq
    \delta}e^{-2\gamma(x-x_{0})^{2k}/h}\,dx\\
    &\leq C_{N}h^{1/2k}.
    \end{aligned}
\end{equation*}
We also have
\begin{equation*}
    \begin{aligned}
    \left|\left|
    \left(
    \begin{matrix}
    ih\chi' & 0\\
    0& -ih\chi'
    \end{matrix}
    \right)
    f
    \right|\right|^{2}_{2}
    &=\int|h\chi'|^{2}e^{-2\imag\vp/h}[|a|^{2}+|b|^{2}]\,
    dx\\
    &\leq
    ch^{2}\int_{\frac{\delta}{2}<|x-x_{0}|<\delta}e^{-2\imag\vp/h}\,
    dx\\
    &\leq
    ch^{2}\int_{\frac{\delta}{2}<|x-x_{0}|<\delta}e^{-2\gamma(x-x_{0})^{2k}/h}\,
    dx\\
    &\leq C'_{N}e^{-\alpha/h}h^{2}
    \end{aligned}
\end{equation*}
for some $\alpha>0$. Hence
$$||(P-\lambda I)\tilde{f}||_{2}\leq C_{N}h^{N+1+\frac{1}{4k}}\quad\forall 0<h<1,$$
where $C_{N}$ is independent of $h$.

\sgap

The last step is to bound $\tilde{f}$ from below:
\begin{equation*}
    \begin{aligned}
    ||\tilde{f}||^{2}_{2}&=\int\chi^{2}|f|^{2}\, dx\\
    &\geq \int_{|x-x_{0}|\leq\frac{\delta}{2}}|f|^{2}\, dx\\
    &=\int_{|x-x_{0}|\leq\frac{\delta}{2}}e^{-2\imag\vp/h}[|a|^{2}+|b|^{2}]\,dx\\
    &\geq c\int_{|x-x_{0}|\leq\frac{\delta}{2}}
    e^{-6\gamma(x-x_{0})^{2k}/h}\, dx\\
    &\geq c_{0}h^{1/2k},
    \end{aligned}
\end{equation*}
where we have used the fact that we have non-zero solutions to the
transport equations.

\sgap

We can now take $$u_{N}:=\tilde{f}/||\tilde{f}||_{2}$$ to conclude the proof of the theorem.

\mgap

\section{Upper Bounds for the Resolvent}\label{S:UpperResolv}

To obtain upper bounds for the resolvent, we will use the
semiclassical pseudodifferential calculus. In this and the
following sections, we will restrict ourselves to $S\in
C^{\infty}(\Rbb;\Rbb)$ such that $S'\in
C^{\infty}_{b}(\Rbb;\Rbb)$, where
\begin{equation*}
    C^{\infty}_{b}:=\{f\in C^{\infty};\, \partial^{\alpha}\! f\in L^{\infty} \, \forall
    \alpha\}.
\end{equation*}
And we will take $A\in\mathcal{S}(\Rbb;\Rbb)$, $A>0$, as before.

In studying our matrix-valued symbols, we might as well use the
norm
$$||B||=\max_{i,j}|b_{ij}|, \qquad\text{where }B=(b_{ij})_{1\leq i,j\leq
n}.$$ Then for our symbol
\begin{equation*}
    p(x,\xi)=
    \left(
    \begin{matrix}
        -\xi-\frac{1}{2}S'(x) &-iA(x)\\
        -iA(x) &\xi-\frac{1}{2}S'(x)
    \end{matrix}
    \right)
\end{equation*}
we have
\begin{equation*}
    \begin{aligned}
    &||p(x,\xi)-\lambda I||\leq C(1+|\xi|),\\
    &||\partial^{\alpha}_{x}p(x,\xi)||\leq C_{\alpha} &&\text{for
    }\alpha\geq 1,\\
    &||\partial_{\xi}p(x,\xi)||=1, &&\text{and}\\
    &||\partial^{\alpha}_{x}\partial^{\beta}_{\xi}p(x,\xi)||=0
    &&\text{for }\beta\geq 1 \text{ and }\alpha+\beta\geq 2.
    \end{aligned}
\end{equation*}
So, in the terminology of Zworski \cite{R:ZworskiSCA}, with the admissible weight function $m(x,\xi)=1+|\xi|$ we
have $p-\lambda I\in S(m)$.

\sgap

With $d(x,\xi)=\det(p(x,\xi)-\lambda I)$ as in the introduction, we now prove an ellipticity result:

\sgap

\begin{lemma}
    Suppose that $A\in\mathcal{S}(\Rbb;\Rbb)$, $A>0$, and that $S$
    is such that $S'\in C^{\infty}_{b}(\Rbb;\Rbb)$.
    If $\lambda\in\Cbb$ is such that $|d(x,\xi)|\geq \epsilon$ for all
    $(x,\xi)\in \Rbb^{2}$,
    for some fixed $\epsilon>0$, then we have
    $$||(P-\lambda I)^{-1}u||_{L^{2}}\leq C(\epsilon,\lambda)||u||_{L^{2}}.$$
\end{lemma}

\begin{proof}
The hypothesis says precisely that
\begin{equation}\label{E:detpgeqeps}
    \begin{aligned}
    d(x,\xi)&\equiv|\det (p(x,\xi)-\lambda I)|\\
    &=|-\xi^{2}+(\tfrac{1}{2}S'(x)+\lambda)^{2}+A(x)^{2}|\\
    &\geq \epsilon \qquad \forall (x,\xi), \qquad \text{for some $\epsilon>0$}
    \end{aligned}
\end{equation}
(which requires $\imag \lambda\neq 0$; also see Section
\ref{S:geommeaning}).

We first demonstrate the ellipticity of the symbol
\begin{equation*}
    (p(x,\xi)-\lambda I)^{-1}=\frac{1}{-\xi^{2}+(\frac{1}{2}S'(x)+\lambda)^{2}+A(x)^{2}}
    \left(
    \begin{matrix}
        \xi-\frac{1}{2}S'(x)-\lambda &iA(x)\\
        iA(x) &-\xi-\frac{1}{2}S'(x)-\lambda
    \end{matrix}
    \right).
\end{equation*}
That is, first we show that
\begin{equation*}
    ||(p(x,\xi)-\lambda I)^{-1}||\leq C(\epsilon,\lambda)(1+|\xi|)^{-1}\qquad
    \forall (x,\xi).
\end{equation*}
For this we let $K>>1$, its precise value to be determined. In fact, we immediately
take $K$ such that $|\lambda|\leq \frac{1}{2}K$. If $|\xi|\leq K$,
then clearly $$||(p(x,\xi)-\lambda I)^{-1}||\leq
\frac{C(K)}{\epsilon}.$$ 
On the other hand, if $|\xi|\geq K$, then
\begin{equation*}
    \begin{aligned}
    |d(x,\xi)|
    &\geq \xi^{2}-|\lambda|^{2}-|S'(x)||\lambda|-(\tfrac{1}{2}S'(x))^{2}-A(x)^{2}\\
    &\geq \tfrac{1}{2}\xi^{2}+\tfrac{1}{2}K^{2}-\tfrac{1}{4}K^{2}-\tfrac{1}{2}|S'(x)|K
    -(\tfrac{1}{2}S'(x))^{2}-A(x)^{2}\\
    &\geq \tfrac{1}{2}\xi^{2} \qquad\text{when $K$ is sufficiently large.}
    \end{aligned}
\end{equation*}
Hence $$||(p(x,\xi)-\lambda I)^{-1}||\leq \frac{C}{|\xi|}.$$

\sgap

Moreover, it is now easy to see that $(p(x,\xi)-\lambda I)^{-1}\in
S(\frac{1}{m})$. Hence, using the pseudodifferential calculus (see Theorem~4.23 of \cite{R:ZworskiSCA}),
$$(P-\lambda I)^{-1}: \quad L^{2}(\Rbb;\Cbb^{2})\rightarrow
L^{2}(\Rbb;\Cbb^{2})$$ is a bounded operator; that is,
$$||(P-\lambda I)^{-1}u||_{L^{2}}\leq C(\epsilon,\lambda)||u||_{L^{2}}.$$

\end{proof}

\mgap

\section{The Geometric Meaning of $|\det (p(x,\xi)-\lambda I)|\geq
\epsilon$}\label{S:geommeaning}

In this section we give a simple geometric meaning to
(\ref{E:detpgeqeps}), as seen in the $\lambda$-plane. That is, we fix $\lambda\in\Cbb$ and suppose that
\begin{equation}\label{E:LBfordet}
    \begin{aligned}
    |\det
    (p(x,\xi)-\lambda I)|^{2}&=\left(-\xi^{2}+(\tfrac{1}{2}S'(x)+\real\lambda)^{2}-(\imag\lambda)^{2}+A(x)^{2}\right)^{2}\\
    &\qquad\qquad+4(\imag\lambda)^{2}(\tfrac{1}{2}S'(x)+\real\lambda)^{2}\\
    &\geq \epsilon^{2} \qquad \forall (x,\xi), \qquad \text{for some $\epsilon>0$.}
    \end{aligned}
\end{equation}
Clearly for this to be true we need $\imag\lambda\neq 0$.

\sgap

\begin{lemma}\label{L:meaning}
    Suppose that $A\in\mathcal{S}(\Rbb;\Rbb)$, $A>0$, and that $S$
    is such that $S'\in C^{\infty}_{b}(\Rbb;\Rbb)$. Then failure of condition (\ref{E:LBfordet}) is equivalent to:
    \begin{equation*}
    	\begin{cases}
		\imag\lambda=0 \\
		\text{OR}\\
		\imag\lambda\neq 0\quad\text{and}\quad \exists x_{0}\in\Rbb\,\,\text{such that}\,\,\real\lambda=-\tfrac{1}{2}S'(x_{0}) \,\,\text{and}\,\,|\imag\lambda|\leq A(x_{0}).
	\end{cases}
    \end{equation*}
\end{lemma}

\sgap

\begin{proof}

If condition (\ref{E:LBfordet}) fails, then either (i) $\imag \lambda=0$ or (ii) $\imag\lambda\neq 0$ and there exists a sequence $(x_{n},\xi_{n})\in\Rbb^{2}$ such that
\begin{equation*}
	\begin{cases}
		-\tfrac{1}{2}S'(x_{n})\to\real\lambda &\text{and}\\
		A(x_{n})^{2}-\xi_{n}^{2}\to(\imag\lambda)^{2}.
	\end{cases}
\end{equation*}
We note that this is impossible if the sequence $(x_{n})_{n=1}^{\infty}$ is unbounded, since $A$ decays to $0$ at infinity. Thus by the Bolzano-Weierstrass theorem we may assume that $\lim x_{n}=x_{0}\in\Rbb$. We thus have a point $x_{0}\in\Rbb$ and a sequence $(\xi_{n})_{n=1}^{\infty}$ such that
\begin{equation*}
	\begin{cases}
		-\tfrac{1}{2}S'(x_{0})=\real\lambda &\text{and}\\
		A(x_{0})^{2}-\xi_{n}^{2}\to(\imag\lambda)^{2}.
	\end{cases}
\end{equation*}
That is,
$$0\leq \lim \xi_{n}^{2}=A(x_{0})^{2}-(\imag\lambda)^{2},$$
so we have a point $x_{0}\in\Rbb$ such that 
\begin{equation*}
	\begin{cases}
		-\tfrac{1}{2}S'(x_{0})=\real\lambda &\text{and}\\
		A(x_{0})^{2}\geq (\imag\lambda)^{2}.
	\end{cases}
\end{equation*}

To prove the other direction, if $\imag\lambda=0$ then clearly condition (\ref{E:LBfordet}) fails, so we assume that there exists some $x_{0}\in\Rbb$ such that $\real\lambda=-\frac{1}{2}S'(x_{0})$ and $|\imag\lambda|\leq A(x_{0})$. Thus
\begin{equation*}
	|\det(p(x_{0},\xi)-\lambda I)|^{2}
	=[-\xi^{2}+(A(x_{0})^{2}-(\imag\lambda)^{2})]^{2},
\end{equation*}
and we choose $\xi\in\Rbb$ such that $\xi^{2}=A(x_{0})^{2}-(\imag\lambda)^{2}$ to see that condition (\ref{E:LBfordet}) fails.

\end{proof}

\sgap

\noindent{\textbf{Example.}} One special case of considerable interest occurs when
$$A(x)=S(x)=\sech(2x).$$ Numerical studies of the
eigenvalues of the Zakharov-Shabat system in this case can be
found in the works of Bronski and Miller \cite{R:Bronski96},
\cite{R:Bronski01}, \cite{R:Miller}. For this example we have the curve
$$\gamma(x)=(-\tfrac{1}{2}S'(x),A(x))=(\tanh(2x)\sech(2x),\sech(2x))=:(\xi,\eta),$$
which is a ``vertical'' lemniscate of Gerono (a.k.a. lemniscate of Huygens), satisfying
$\xi^{2}=\eta^{2}(1-\eta^{2})$, as seen in Figure~\ref{F:Gerono}. The set of $\lambda$ such that condition (\ref{E:LBfordet}) fails is precisely (the convex hull of the lemniscate)$\cup\Rbb$.

\begin{figure}
\begin{center}
\epsfig{file=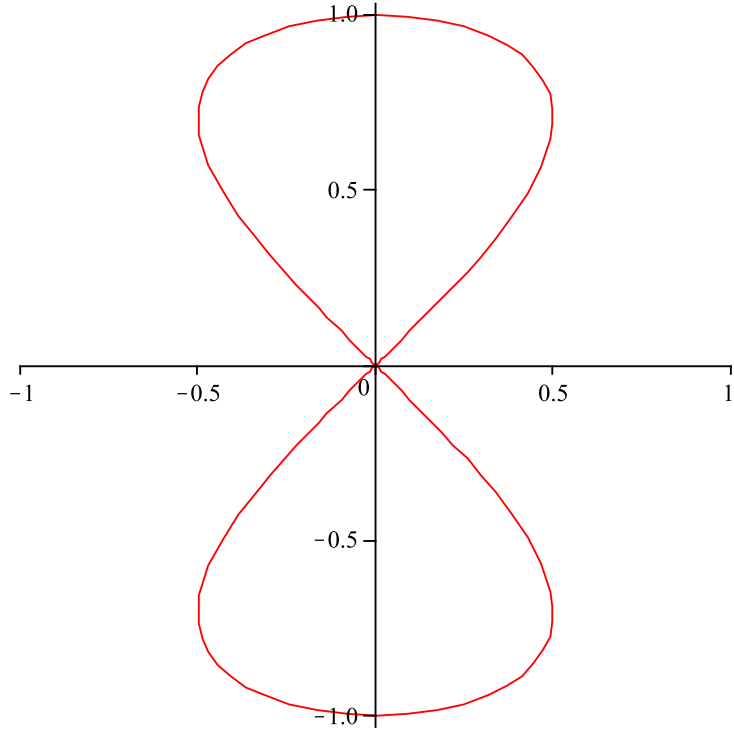,height=10cm}
\end{center}
\caption{The Lemniscate of Gerono.}
\label{F:Gerono}
\end{figure}

\sgap

\noindent\textbf{Example.} Another interesting special case occurs when
$$A(x)=-\sech x \quad\text{and}\quad S'(x)=-\mu\tanh x,$$
where $\mu$ is a real parameter. The semiclassical limit of the
Zakharov-Shabat eigenvalue problem in this case was studied by
Tovbis and Venakides, who found an explicit solution
\cite{R:TovbisVenakides}. Here the turning point curve is
simply the ellipse given in the $(\xi,\eta)$-plane by
$$\left(\frac{2\xi}{\mu}\right)^{2}+\eta^{2}=1.$$

\mgap

\section{$\Sigma(p)$ and $\Sigma_{\infty}(p)$}\label{S:Sigmas}

In the next sections we use the methods of Dencker, Sj\"{o}strand,
and Zworski \cite{R:DenckerPSsys}, \cite{R:DSZ}, to prove the
discreteness of the spectrum off the real axis. We begin with two central definitions from \cite{R:DenckerPSsys}:

\begin{definition}
    Let $p\in C^{\infty}(T^{*}\Rbb^{n},\mathcal{L}(\Cbb^{N},\Cbb^{N}))$. We
    denote the closure of the set of eigenvalues of $p$ by:
    \begin{equation*}
        \Sigma(p)=\overline{\{\lambda\in\Cbb;\, \exists w\in T^{*}\Rbb^{n},\,|p(w)-\lambda I|=0\}}
    \end{equation*}
    (with the notation $|T|=\det T$) and the eigenvalues at infinity by:
    \begin{equation*}
        \Sigma_{\infty}(p)=\{\lambda\in\Cbb;\, \exists w_{j}\rightarrow\infty,\, \exists u_{j}\in\Cbb^{N}\backslash 0\,\,\text{such that}\,\,
        |p(w_{j})u_{j}-\lambda u_{j}|\slash |u_{j}|\rightarrow 0,\, j\rightarrow
        \infty\},
    \end{equation*}
    which is closed in $\Cbb$.
\end{definition}

\sgap

The statement that $\lambda(x,\xi)$ is an eigenvalue of the matrix
\begin{equation*}
    p(x,\xi)=
    \left(
    \begin{matrix}
        -\xi-\frac{1}{2}S'(x) &-iA(x)\\
        -iA(x)&\xi-\frac{1}{2}S'(x)
    \end{matrix}
    \right)
\end{equation*}
is equivalent to the statement that
\begin{equation*}
    \begin{cases}
    \imag\lambda=0 &\text{and}\quad
    \xi^{2}=(\frac{1}{2}S'(x)+\real\lambda)^{2}+A(x)^{2}\\
    \text{OR}\\
    \frac{1}{2}S'(x)+\real\lambda=0 &\text{and}\quad
    \xi^{2}=A(x)^{2}-(\imag\lambda)^{2}
    \end{cases}
\end{equation*}
Hence $\Sigma(p)$ is precisely the set
\begin{equation*}
    \{\lambda\in\Cbb;\, \imag\lambda=0\}\cup
    \{\lambda\in\Cbb;\,\exists x\in\Rbb\text{ s.t. }\real\lambda=-\tfrac{1}{2}S'(x)\text{ and }|\imag\lambda|\leq
    A(x)\}.
\end{equation*}
That this set is closed follows from the same argument as in the proof of Lemma~\ref{L:meaning}. Moreover, we see that $\Sigma(p)$ is precisely the set for which condition~(\ref{E:LBfordet}) fails. 

\sgap

We now turn to $\Sigma_{\infty}(p)$ and prove that
$\Sigma_{\infty}(p
)\subset\Rbb$. Let $\lambda\in \Cbb$ be such that
$\imag\lambda\neq 0$. We will show that $\lambda\notin\Sigma_{\infty}(p)$. In the following calculations, we use the
expression of $p(x,\xi)-\lambda I$ as a sum of a selfadjoint
matrix and an anti-selfadjoint matrix:
$$p(x,\xi)-\lambda I=X+Y$$
where
\begin{equation*}
    X=
    \left(
    \begin{matrix}
        -\xi-\frac{1}{2}S'(x)-\real\lambda &0\\
        0&\xi-\frac{1}{2}S'(x)-\real\lambda
    \end{matrix}
    \right),
\end{equation*}
\begin{equation*}
    Y=
    -i\left(
    \begin{matrix}
        \imag\lambda &A(x)\\
        A(x)&\imag\lambda
    \end{matrix}
    \right),
\end{equation*}
and where the commutator is
\begin{equation*}
    [X,Y]
    =XY+\overline{(XY)^{T}}
    =
    2i\xi A(x)
    \left(
    \begin{matrix}
        0 &1\\
        -1&0
    \end{matrix}
    \right).
\end{equation*}
Correspondingly, if we write $\vec{u}=\left(
\begin{matrix}u_{1}\\
u_{2}
\end{matrix}\right)\in\Cbb^{2}$, then
\begin{equation*}
    |X\vec{u}|^{2}=(\xi+\tfrac{1}{2}S'(x)+\real\lambda)^{2}|u_{1}|^{2}
    +(-\xi+\tfrac{1}{2}S'(x)+\real\lambda)^{2}|u_{2}|^{2},
\end{equation*}
\begin{equation*}
    |Y\vec{u}|^{2}=[A(x)^{2}+(\imag\lambda)^{2}][|u_{1}|^{2}+|u_{2}|^{2}]
    +4A(x)(\imag\lambda)\real(\overline{u_{1}}u_{2}),
\end{equation*}
and, taking the convention $\langle a,b\rangle=\overline{a}b$,
\begin{equation*}
    \langle [X,Y]\vec{u},\vec{u}\rangle =-4\xi
    A(x)\imag(\overline{u_{1}}u_{2}).
\end{equation*}

\sgap

To prove that $\lambda\notin \Sigma_{\infty}(p)$, we consider
\begin{equation*}
    \{(x,\xi);\, |x|\geq C\}\cup \{(x,\xi);\, |\xi|\geq C\},
\end{equation*}
for $C>0$ to be determined.

\sgap

In the first case, we take $|\xi|\geq R$, where $R$ is to be
determined, depending only on $||S'||_{\infty}$,
$||A||_{\infty}$, and $\lambda$. We then have
\begin{equation*}
    |\pm\xi +\tfrac{1}{2}S'(x)+\real\lambda|
    \geq |\xi|-|\tfrac{1}{2}S'(x)+\real\lambda|.
\end{equation*}
Hence
\begin{align*}
    |X\vec{u}|^{2}&\geq\left[|\xi|-|\tfrac{1}{2}S'(x)+\real\lambda|\right]^{2}|\vec{u}|^{2}\\
    &\geq\left[\tfrac{1}{2}|\xi|+\tfrac{1}{2}R-|\tfrac{1}{2}S'(x)+\real\lambda|\right]^{2}|\vec{u}|^{2}\\
    &\geq\tfrac{1}{4}|\xi|^{2}|\vec{u}|^{2}
\end{align*}
when $R$ is large enough. Also,
\begin{equation*}
    \langle[X,Y]\vec{u},\vec{u}\rangle
    \geq -2|\xi|A(x)|\vec{u}|^{2},
\end{equation*}
so that
\begin{equation*}
    \begin{aligned}
    |X\vec{u}|^{2}+\langle[X,Y]\vec{u},\vec{u}\rangle
    &\geq |\xi|\left[\tfrac{1}{4}|\xi|-2A(x)\right]|\vec{u}|^{2}\\
    &\geq R\left[\tfrac{1}{4}R-2A(x)\right]|\vec{u}|^{2}.
    \end{aligned}
\end{equation*}
Taking $R\geq 4+8||A||_{\infty}$, we have
$$|(p-\lambda I)\vec{u}|^{2} \geq R|\vec{u}|^{2}.$$

\sgap

In the second case, $\xi$ is bounded: $|\xi|\leq R$. Let
\begin{equation*}
    \epsilon = \min\{ \tfrac{1}{8}|\imag\lambda|,\, \tfrac{1}{8R}|\imag\lambda|^{2}
    \}\quad (>0).
\end{equation*}
We then take $C>0$ to be so large that $A(x)\leq\epsilon$ for all
$|x|\geq C$. Then we have
\begin{equation*}
    \begin{aligned}
    |(p-\lambda I)\vec{u}|^{2}
    &\geq |Y\vec{u}|^{2}+\langle [X,Y]\vec{u},\vec{u}\rangle\\
    &\geq \left[ (\imag\lambda)^{2}-2\epsilon(|\imag\lambda|+|\xi|)
    \right]|\vec{u}|^{2}\\
    &\geq \left[ (\imag\lambda)^{2}-\tfrac{1}{4}(\imag\lambda)^{2}-\tfrac{1}{4R}(\imag\lambda)^{2}|\xi|
    \right]|\vec{u}|^{2}\\
    &\geq\tfrac{1}{2}(\imag\lambda)^{2}|\vec{u}|^{2}.
    \end{aligned}
\end{equation*}

\sgap

So in all cases we have $\lambda\notin\Sigma_{\infty}(p)$, proving that $\Sigma_{\infty}(p)\subset\Rbb$.

\mgap

\section{Discreteness of the Spectrum Away From
$\Rbb$}\label{S:Discreteness}

Here we only very slightly modify the methods of Dencker, Sj\"{o}strand, and Zworski (Proposition~2.19 of \cite{R:DenckerPSsys} and
Proposition~3.3 of \cite{R:DSZ}) to prove discreteness of the
spectrum away from the real line.

\sgap

\begin{proposition}
    Suppose that $A\in\mathcal{S}(\Rbb;\Rbb)$, $A>0$, and that $S$
    is such that $S'\in C^{\infty}_{b}(\Rbb;\Rbb)$.
    Let $\Omega\subset\Cbb$ be an open, connected, and bounded set such that
    \begin{equation*}
        \overline{\Omega}\cap
        \Sigma_{\infty}(p)=\emptyset\quad\text{and}\quad
        \Omega\cap\complement\Sigma(p)\neq\emptyset.
    \end{equation*}
    Then $$(P(h)-zI)^{-1},\qquad 0<h<<1,\, z\in\Omega,$$ is a
    meromorphic family of operators with poles of finite rank. In
    particular, for $h$ sufficiently small, the spectrum of
    $P(h):=P(x,hD)$ is discrete in any such set. When
    $\Omega\cap\Sigma(p)=\emptyset$ we find that $\Omega$
    contains no spectrum of $P(h)$.
\end{proposition}

\sgap

\begin{proof}
    We first claim that $\exists C>0$ such that
    \begin{equation}\label{E:bdinv}
        |(p(w)-zI)^{-1}|\leq C\qquad\text{if }z\in\Omega\text{ and
        }|w|>C.
    \end{equation}
    Suppose not. Then $\exists w_{j}\rightarrow
    \infty$ and $z_{j}\in\Omega$ such that
    $$|(p(w_{j})-z_{j}I)^{-1}|\rightarrow\infty\quad\text{as }j\rightarrow\infty.$$
    Thus $\exists u_{j}\in\Cbb^{2}$ with $|u_{j}|=1$ such that
    $$|(p(w_{j})-z_{j}I)u_{j}|\rightarrow 0.$$
    Since $\Omega$ is bounded, we may take a subsequence such that
    $$z_{j}\rightarrow z\in\overline{\Omega}\cap\Sigma_{\infty}(p)=\emptyset$$
    which of course is impossible.

    \sgap

    Next we show that $\exists\lambda_{0}\in\Omega$ such that $(p(w)-\lambda_{0}I)^{-1}\in
    C^{\infty}_{b}$. In fact, let $\lambda_{0}\in \Omega\cap\complement\Sigma(p)$. By the same argument as in the proof of Lemma~\ref{L:meaning}, there exists some $\epsilon>0$ such that 
    $$|-\xi^{2}+(\tfrac{1}{2}S'(x)+\lambda_{0})^{2}+A(x)^{2}|\geq \epsilon\qquad\forall (x,\xi)\in T^{\ast}\Rbb.$$ Then it is easy to see that 
    \begin{equation*}
    	\begin{aligned}
        		(p(x,\xi)-\lambda_{0}I)^{-1}
		&=\frac{1}{-\xi^{2}+(\frac{1}{2}S'(x)+\lambda_{0})^{2}+A(x)^{2}}
        		\left(
        		\begin{matrix}
            		\xi-\frac{1}{2}S'(x)-\lambda_{0} &iA(x)\\
            		iA(x) &-\xi-\frac{1}{2}S'(x)-\lambda_{0}
        		\end{matrix}
        		\right)\\
		&\in C_{b}^{\infty}.
    	\end{aligned}
    \end{equation*}

    \sgap

    We now let
    $\chi\in C^{\infty}_{0}(T^{*}\Rbb)$, $0\leq\chi(w)\leq 1$, and
    $\chi(w)=1$ when $|w|\leq C$, where $C$ is given by
    (\ref{E:bdinv}). Let
    \begin{equation*}
        R(w,z)=\chi(w)(p(w)-\lambda_{0}I)^{-1}+(1-\chi(w))(p(w)-zI)^{-1}
    \end{equation*}
    for $z\in\Omega$, which, by our previous arguments, is in $C^{\infty}_{b}$.
    The semiclassical symbol calculus then gives
    $$R^{w}(x,hD,z)(P(h)-zI)=I+hB_{1}(h,z)+K_{1}(h,z)$$ and
    $$(P(h)-zI)R^{w}(x,hD,z)=I+hB_{2}(h,z)+K_{2}(h,z),$$ where
    $K_{j}(h,z)$ are compact operators on $L^{2}(\Rbb)$ depending
    holomorphically on $z$, vanishing for $z=z_{0}$, and where the
    $B_{j}(h,z)$ are bounded on $L^{2}(\Rbb)$, $j=1,2$. By the
    analytic Fredholm theory we then have that $(P(h)-zI)^{-1}$
    is meromorphic in $z\in\Omega$ for $h$ sufficiently small.
    When $\Omega\cap\Sigma(p)=\emptyset$ we may take
    $R(w,z)=(p(w)-zI)^{-1}$. Then $K_{j}\equiv 0$ and
    $P(h)-zI$ is invertible for small enough $h$.
\end{proof}

\sgap

\sgap

\noindent
\textit{E-mail address:} \textbf{vanvalke@grinnell.edu}

\end{document}